\documentclass[%
 reprint,
superscriptaddress,
 amsmath,amssymb,
 aps,
]{revtex4-2}

\usepackage{graphicx}
\usepackage{dcolumn}
\usepackage{bm}
\usepackage{hyperref}
\usepackage{orcidlink}
\usepackage[ruled,vlined]{algorithm2e}
\usepackage{import}
\usepackage{transparent}
\usepackage{cleveref}
\usepackage{amsthm}
\usepackage{amsmath}
\usepackage{braket}
\usepackage{subfig}
\usepackage{geometry}
\usepackage{adjustbox}
\usepackage{booktabs}
\usepackage{float}

\definecolor{brandblue}{HTML}{1F77B4}
\definecolor{brandbrown}{HTML}{8C564B}
\definecolor{brandteal}{HTML}{4C9085}
\definecolor{brandorange}{HTML}{FF7F0E}
\definecolor{brandpurple}{HTML}{9467bd}
\definecolor{brandgreen}{HTML}{2ca25f}

\definecolor{lightcoral}{RGB}{240, 128, 128}
\definecolor{custompink}{RGB}{232,129,229}
\definecolor{customcyan}{HTML}{47ced0}
\definecolor{customgreen}{RGB}{80,205,60}

\newcommand{\orcid}[1]{\href{https://orcid.org/#1}{\includegraphics[width=0.9em]{orcid.pdf}}}

\newtheorem{theorem}{Theorem}
\newtheorem*{restatement}{Restatement of Theorem~\ref{thm:detailed_balance_mcmc}}

\begin{document}

\preprint{APS/123-QED}

\title{Revisiting the Quantum-Guided Cluster Algorithm: Improvements and Numerical Experiments}

\author{Peter J. Eder\,\orcidlink{0009-0006-3244-875X}}
\email{peter-josef.eder@tum.de}
\affiliation{Technical University of Munich (TUM), Garching, Germany}
\affiliation{Siemens AG, Germany}

\author{Sarah Braun\,\orcidlink{0000-0002-7032-6116}}
\email{sarah.braun@siemens.com}
\affiliation{Siemens AG, Germany}

\begin{abstract}
We study correlation-guided cluster algorithms for solving the \textsc{Max-Cut} problem that iteratively try to improve solutions by updating clusters of nodes. Building on the recently proposed quantum-guided cluster algorithm (QGCA)~\cite{Eder2025-Cluster}, which leverages precomputed two-point correlations to guide collective updates, we extend the cluster construction by incorporating next-nearest-neighbor (NNN) information. We evaluate this extension across different correlation sources on random regular graphs and non-degenerate tile-planted instances. Notably, we observe particularly strong performance on non-degenerate instances and provide a scaling analysis for this class. Finally, we outline an extension toward a correlation-guided Markov-chain Monte Carlo algorithm, whose detailed analysis remains an open direction for future work.
\end{abstract}

\maketitle

\section{Introduction}
\label{sec:introduction}

Exploring Ising spin glasses and related combinatorial optimization problems, such as \textsc{Max-Cut}, remains computationally challenging due to their highly rugged energy landscapes induced by disorder and frustration. Classical approaches based on Monte Carlo methods~\cite{Metropolis1953,Hastings1970,Geman1984}, such as simulated annealing (SA), provide versatile frameworks for sampling and optimization, but often struggle to efficiently reach low-energy states in such complex landscapes.

Cluster algorithms (CAs), such as those introduced by Swendsen and Wang~\cite{Swendsen1987} and Wolff~\cite{Wolff1989}, accelerate convergence in ferromagnetic systems by performing collective spin updates. Extensions to spin glasses, including Houdayer's algorithm~\cite{Houdayer2001}, have demonstrated improvements in certain settings. However, their effectiveness is typically limited for generic spin-glass instances due to percolation effects, where clusters grow to system-spanning size and thereby lose their ability to induce meaningful transitions~\cite{Radicchi2015,Kessler1990}.

To address these limitations, we recently introduced the quantum-guided cluster algorithm (QGCA)~\cite{Eder2025-Cluster}, which leverages precomputed two-point correlations, obtained from classical or quantum methods, to guide cluster formation. By incorporating such information, the algorithm enables collective updates that reflect the underlying structure of low-energy configurations, thereby facilitating more effective exploration of the configuration space. We also successfully applied the QGCA as a post-processing method for the Dial-a-Ride Problem with time windows in a hybrid quantum-classical framework~\cite{Eder2026}.

In this work, we extend the QGCA and perform a more detailed analysis. Our contributions are as follows:
\begin{itemize}
    \item We extend the QGCA by incorporating next-nearest neighbor (NNN) information into the cluster-building procedure and analyze its performance across different correlation sources on random regular \textsc{Max-Cut} instances, leading to improved performance.
    
    \item We study non-degenerate tile-planted Chook instances and show that the algorithm exhibits particularly strong performance on this class, together with a scaling analysis.
    
    \item As an outlook, we present an extension of the method to a correlation-guided Markov-chain Monte Carlo (MCMC) algorithm, whose detailed numerical analysis is left for future work.
\end{itemize}

\section{The \textsc{Max-Cut} Problem \& Ising Spin Glasses}

The \textsc{Max-Cut} problem is a fundamental NP-hard combinatorial optimization problem~\cite{Karp1972}. Even for moderately sized and dense graphs, determining optimal solutions can become computationally prohibitive for classical algorithms~\cite{Charfreitag2022}. As its decision version is NP complete, \textsc{Max-Cut} serves as a canonical formulation for a wide range of NP-hard problems, including applications in areas such as scheduling and logistics~\cite{Lucas2014}.

Formally, a weighted \textsc{Max-Cut} instance is defined on an undirected graph $G = (V, E)$ with vertex set $V$ of size $n$ and edge set $E$. Each edge $\{i,j\} \in E$ is associated with a weight $A_{ij} \in \mathbb{R}$. The objective is to partition the vertices into two disjoint subsets such that the total weight of edges connecting the two subsets is maximized.

This objective can be expressed as
\begin{equation}
C(x) = \frac{1}{2} \sum_{\{i,j\} \in E} A_{ij} \left(1 - x_i x_j\right),
\label{eq:max_cut_objective_function}
\end{equation}
where $x_i \in \{-1,1\}$ encodes the assignment of vertex $i$ to one of the two subsets.

The problem admits an equivalent formulation in terms of Ising spin glasses, with Hamiltonian
\begin{equation}
H(x) = - \sum_{\{i,j\} \in E} J_{ij} x_i x_j,
\label{eq:spin_glass}
\end{equation}
where the couplings $J_{ij}$ represent interaction strengths between spins. Minimizing $H$ is equivalent to maximizing~\eqref{eq:max_cut_objective_function}.

\section{Algorithmic Design}
\label{sec:algorithmic_design}

\subsection{Nearest-Neighbor Variant}

We begin by briefly recalling the nearest-neighbor variant of the algorithm introduced in~\cite{Eder2025-Cluster}, before extending it to include NNNs in the subsequent section. The proposed CA follows the general structure of SA, but replaces single-spin updates with collective flips of spin clusters. The full procedure is given in Algorithm~\ref{alg:qcga_linear_schedule}. We refer to this method as a \emph{correlation-guided cluster algorithm}, or \emph{quantum-guided} when quantum-derived correlations are used.

\begin{algorithm}[htbp]
\caption{Correlation-Guided CA, similar to \cite[Algorithm~1]{Eder2025-Cluster}. \textcopyright~2025 IEEE.}
\label{alg:qcga_linear_schedule}
\SetAlgoLined
\KwIn{Hamiltonian $H$, correlation matrix $Z$, final inverse temperature $\beta_f$, number of iterations $i_{\text{max}}$, and probability scaling factor $\ell_{\text{scale}}$.}
\KwOut{Best energy $E_{\text{opt}}$ \& corresponding state vector $x$.}

$\beta \leftarrow 0$\;
Initialize random spin configuration $x$\;
$E_{\text{opt}} \leftarrow H(x)$\;
\While{$\beta < \beta_f$}{
    $i \leftarrow$ Uniform$(1, n)$; \tcp*[f]{Choose seed node}\\
    $\mathcal{C} \gets$ \texttt{CreateCluster}$(i, x, H, Z, \ell_{\text{scale}})$\;
    $x' \gets$ Flip$(x, \mathcal{C})$; \tcp*[f]{Flip cluster}\\
    $\Delta E \leftarrow H(x') - H(x)$\;
    \If{$\Delta E \leq 0$}{
        $x \leftarrow x'$\;
        \If{$H(x) < E_{\mathrm{opt}}$}{
            $E_{\text{opt}} \leftarrow H(x)$\;
        }
    }
    \ElseIf{$\mathrm{Uniform}(0, 1) < e^{-\beta \Delta E}$}{
        $x \leftarrow x'$; \tcp*[f]{Metropolis criterion}
    }
    $\Delta \beta \leftarrow \beta_f/(i_{\text{max}} - 1) |C|$\;
    $\beta \leftarrow \beta + \Delta \beta$; \tcp*[f]{Reduce temperature}\\
}
\end{algorithm}

The algorithm operates on a precomputed correlation matrix $Z$, which is evaluated once prior to the optimization. This matrix determines how clusters are formed and thereby guides the exploration of the configuration space. In particular, two-point correlations $Z_{i\neq j} \in [-1,1]$ quantify whether spins $i$ and $j$ are likely to be aligned (values close to $1$) or anti-aligned (values close to $-1$) in a \textsc{Max-Cut} solution. The framework is agnostic to the origin of these correlations and can incorporate both classical and quantum sources.

Starting from a random initial configuration $x$, each iteration selects a seed node uniformly at random and constructs a cluster $\mathcal{C}$ based on the correlations encoded in $Z$. The cluster is then flipped, and the resulting energy change $\Delta E$ is computed. The update is accepted according to the Metropolis rule
\begin{equation}
    p_{\text{acc}}(x \to x') = \min\!\left(1, e^{-\beta \Delta E}\right),
    \label{eq:metropolis_criterion}
\end{equation}
while the inverse temperature $\beta$ is gradually increased until a final value $\beta_f$ is reached.

The annealing schedule is defined such that $\beta_f$ is attained after $i_{\text{max}}$ effective single-spin updates, treating a cluster flip of size $k$ as equivalent to $k$ individual updates. This ensures comparability with SA independent of implementation details. Unlike SA, where the acceptance probability is evaluated for every spin flip, the present approach performs this evaluation only once per cluster update, which slightly penalizes the method.

Cluster construction is carried out by the procedure \texttt{CreateCluster()}, which uses the correlation matrix $Z$ to counteract percolation effects that typically limit CAs in spin-glass systems~\cite{Kessler1990}. Beginning from the seed node, neighboring nodes are processed in random order and added to the cluster with probability
\begin{equation}
p^{(i,j)}_{\text{link}}(x) := \min\left(1,\max\left(0, \frac{\ell_{\text{scale}}}{\ell_{\text{perc}}} x_i x_j Z_{ij}\right)\right),
\label{eq:p_link}
\end{equation}
where $\ell_{\text{scale}}>0$ is a tunable parameter and
\begin{equation}
\ell_{\text{perc}} := \frac{\langle d \rangle}{2 \mathbb{E}[|Z_{i\neq j}| \mid Z_{i\neq j} \neq 0] \cdot (\langle d^2 \rangle - \langle d \rangle)}
\end{equation}
provides an estimate of the percolation threshold. Here, $\langle d \rangle$ and $\langle d^2 \rangle$ denote the first and second moments of the degree distribution, while $\mathbb{E}[|Z_{ij}| \mid Z_{ij} \neq 0]$ captures the average magnitude of non-zero correlations. This expression is based on~\cite{Radicchi2015} and adjusted to account for the scale of the correlations. The cluster building procedure for the extension of NNNs introduced in the following section is depicted in \Cref{fig:cluster_building}.

\begin{figure*}[htbp]
  \centering
  \adjustbox{max width=\textwidth, scale=1.1}{%
    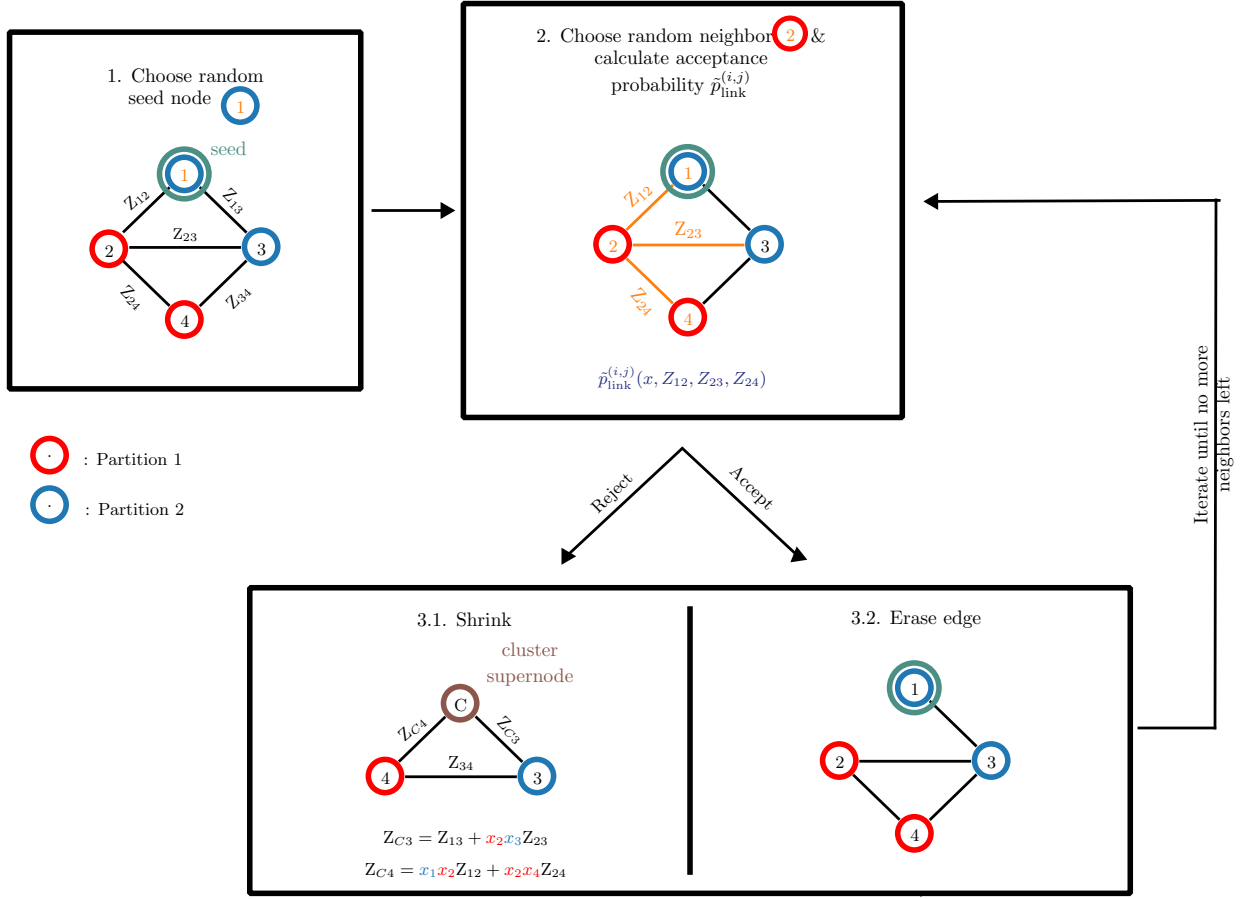%
  }
  \caption{Illustration of the cluster-building process, characterizing the pseudocode function \texttt{CreateCluster()} for NNNs, similar to \cite[Figure~3]{Eder2025-Cluster}. Starting from a randomly selected seed node (here, $i=1$), neighboring vertices are iteratively considered for inclusion in the cluster with probability $\Tilde{p}_{\text{link}}^{(i,j)}$ from \Cref{eq:p_link_nnn}, determined by the correlation matrix $Z$. Accepted nodes undergo a shrinking step before further expansion, while rejected nodes have their edge to the cluster removed. \textcopyright~2025 IEEE.}
  \label{fig:cluster_building}
\end{figure*}

During the construction process, nodes are either rejected, removing the corresponding edge, or accepted and merged into the cluster via a shrinking step, forming a supernode. For details on the shrinking step, see, for example, ~\cite{Bravyi2020,Fischer2024}. The procedure terminates once no additional neighbors are left. Note that these modifications on the graph are reset after each cluster construction.

\clearpage

\subsection{Extension to Next-Nearest Neighbors}
\label{sec:extension_to_next_nearest_neighbors}

To further enhance convergence, we extend the cluster construction to include \emph{next-nearest neighbors} (NNN). In this variant, the link probability incorporates not only the direct correlation between spins $i$ and $j$, but also contributions from the neighbors of $j$. The modified link probability is defined as
\begin{align}
\tilde{p}_{\text{link}}^{(i,j)}(x)
&:= \min\Bigg\{1,\ 
    \max\Bigg[0,\ 
        e\frac{\ell_{\text{scale}}}{\ell_{\text{perc}}}
        x_i x_j Z_{ij} \nonumber\\
&\qquad - (1-e)\frac{\ell_{\text{scale}}}{\ell_{\text{perc}}
        }\sum_{k \in \mathcal{N}(j)} x_j x_k Z_{jk}
\Bigg]
\Bigg\},
\label{eq:p_link_nnn}
\end{align}
where $e \in [0,1]$ controls the relative contribution of direct neighbor and NNN terms, and $\mathcal{N}(j)$ denotes the set of neighbors of node $j$. This extension introduces an additional parameter that must be tuned, but simultaneously incorporates more structural information into the cluster-building process, which can improve the ability of the algorithm to escape local minima.

Including NNNs also affects the annealing schedule. Instead of scaling the inverse temperature increment solely with the cluster size, we account for the total number $K$ of nodes considered during cluster construction, including both accepted nodes and rejected neighboring candidates. The update rule becomes
\begin{equation}
    \Delta \beta = \frac{\beta_f}{i_{\text{max}}-1} \, K,
\end{equation}
which maintains consistency with SA. The original scaling can be recovered by replacing $K$ with the cluster size $|\mathcal{C}|$, thereby ensuring a consistent comparison in computational cost through an appropriate implementation choice. This variant is described in Appendix~\ref{apx:alternative_cost_calculations_for_nnn} but is not used in the numerical results for simplicity.

\section{Numerical Experiments with Next-Nearest Neighbors}
\label{sec:numerical_experiments_with_next_nearest_neighbors}

In the proposed correlation-guided CA, clusters are constructed using precomputed information derived from various sources, including coupling constants, semidefinite programming (SDP), thermal Monte Carlo sampling, and the Quantum Approximate Optimization Algorithm (QAOA). These correlation types and their computation have been described in detail in~\cite{Eder2025-Cluster} and are therefore not repeated here.

Furthermore, in~\cite{Eder2025-Cluster}, we benchmarked the correlation-guided CA on random regular \textsc{Max-Cut} instances of varying degrees, comparing it to SA and analyzing its performance when guided by the correlations mentioned above. The results demonstrated that the effectiveness of the method strongly depends on the quality of the correlations and the level of frustration, with improved correlations leading to significant performance gains, especially for more frustrated graphs. In this section, we analyze the same correlation types, while considering NNN interactions.

In all experiments, edge weights are drawn uniformly at random from $\{-1,1\}$, except for the Chook instances discussed in \Cref{sec:non_degenerate_chook_problems}. The final inverse temperature is set to $\beta_f = 8$ for both the correlation-guided CA and SA. Instances of size $n=28$ and $3$-regular graphs with $n=100$ are solved to optimality using Gurobi~\cite{gurobi}, whereas $20$- and $40$-regular graphs of size $n=100$ are solved using SA with $10{,}000n$ iterations and $50$ repetitions. Accordingly, approximation ratios $\bar{C}_{\text{app}}$ are computed with respect to optimal solutions for the smaller and $3$-regular instances, and with respect to the best solutions found for the $20$- and $40$-regular graphs.

\subsection{Coupling Constants, Semidefinite Programming \& Monte Carlo Correlations}
\label{sec:coupling_constants_semidefinite_programming_monte_carlo_correlations}

We begin by analyzing the CA guided by classical correlations. \Cref{fig:MC_NNN} shows the percentage of optimal solutions found as a function of the number of iterations. The hyperparameters are determined via a grid search over selected values of $\ell_{\text{scale}}$ within $[1,1000]$ and $e \in \{0.5,0.6,\dots,1.0\}$, where each configuration is evaluated using 50 mutations with $50n$ iterations, and the best-performing setting is selected for each data point. The graph ensembles for degrees three and 20 are the same as the ones used in~\cite{Eder2025-Cluster}, and we additionally consider ten random regular graphs of degree 40 in \Cref{fig:MC_degree_40_combo_NNN} to probe more strongly frustrated instances.

\begin{figure*}[htbp]
    \centering
    \subfloat[3-regular (left) and 20-regular (right) degrees]{%
    \input{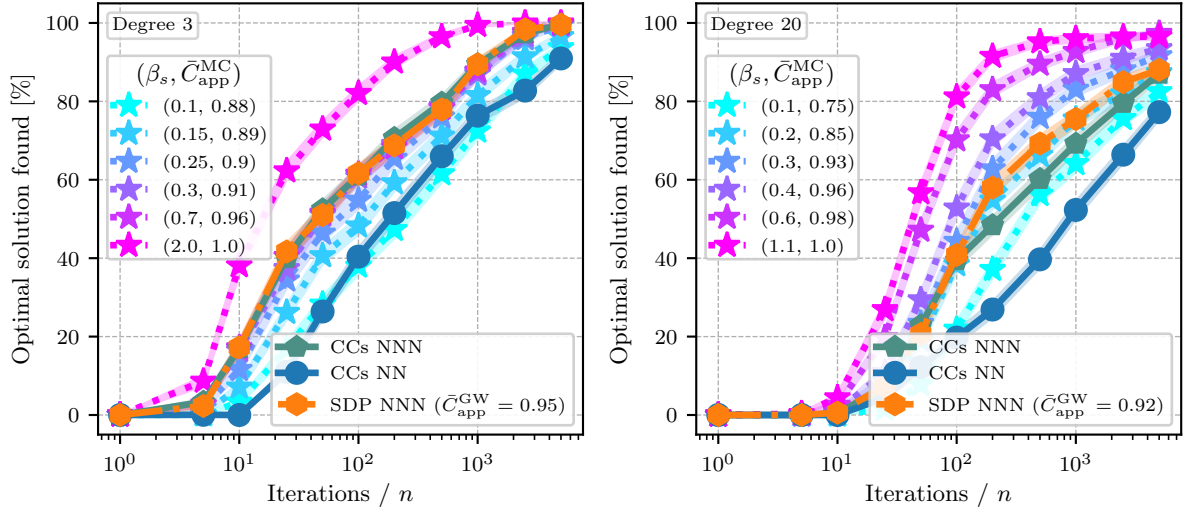}%
        \label{fig:MC_degrees_3_20_NNN}%
    }\par
    \subfloat[40-regular (left) with its frustration graph (right)]{%
    \input{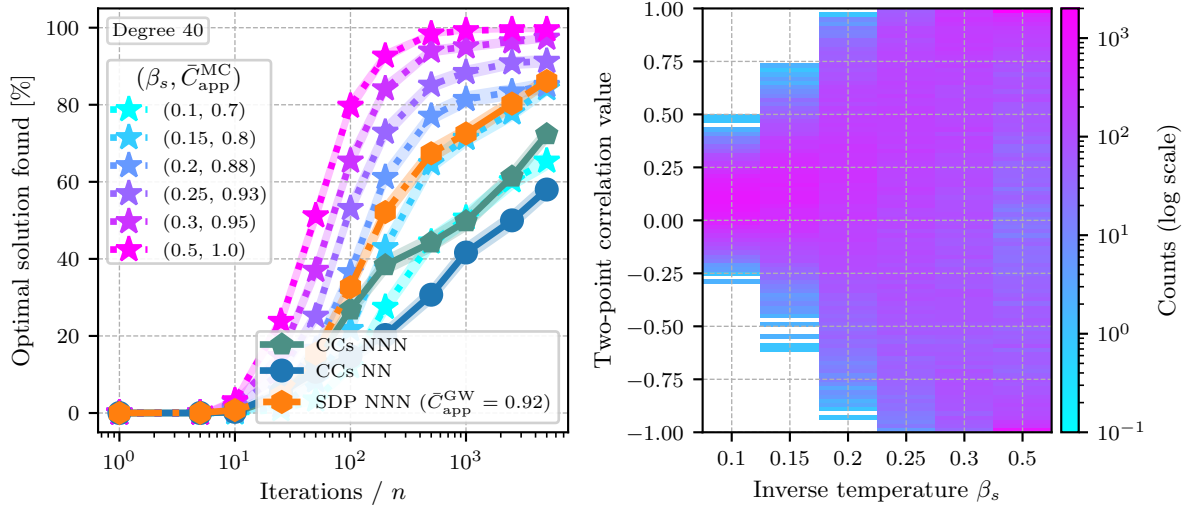}%
        \label{fig:MC_degree_40_combo_NNN}%
    }
    \caption{The plots in (a) and the left panel of (b) show the percentage of optimal solutions found as a function of the number of iterations for the CA considering next-nearest neighbors (NNN) guided by various correlations, as well as nearest neighbors (NN) guided by coupling constants (CCs). In (a), the same graph instances as in~\cite{Eder2025-Cluster} are used, while in (b) ten random $40$-regular graphs of size $n=100$ are considered. Together, these subfigures illustrate the performance of the CA guided by CCs, semidefinite programming (SDP) correlations, and thermal correlations sampled from Monte Carlo (MC) simulations at different inverse temperatures $\beta_s$ for NNN interactions. The right panel of (b) shows the corresponding $Z_{ij}^{\text{MC}}$ values for edges with $J_{ij}=1$, obtained from MC sampling and grouped into bins for degree 40 at the same inverse temperatures as in the left plot.}
    \label{fig:MC_NNN}
\end{figure*}

The results show that incorporating NNNs improves performance over the nearest-neighbor variant when using coupling constants. Moreover, the version guided by SDP correlations outperforms the coupling-constant-based variant, with comparable performance only for degree-three graphs. Thermal correlations further improve performance for sufficiently high sample quality, with stronger effects for higher-degree graphs. In particular, for degree-three graphs, improvements occur only for near-optimal samples, whereas for larger degrees even moderately accurate correlations provide benefits. Overall, these findings indicate that the advantage of more informative correlations increases with frustration.

The frustration plot (i.e., $Z_{ij}^{\text{MC}}$ for $J_{ij}=1$) for the degree-40 graphs is shown on the right-hand side of \Cref{fig:MC_degree_40_combo_NNN} and is obtained by binning the correlation values from MC sampling. In the absence of frustration, these correlations cluster near $Z_{ij}=1$, while a broader distribution extending towards negative values reflects competing interactions and thus higher frustration. The observed spread confirms the strong frustration of these instances and shows that improved thermal samples capture more informative structure. Consistently, thermal correlations outperform SDP correlations at comparable solution qualities; for example, $(\beta_s, \Bar{C}_{\text{app}}^{\text{MC}}) = (0.3, 0.93)$ for degree-20 graphs and $(\beta_s, \Bar{C}_{\text{app}}^{\text{MC}}) = (0.25, 0.93)$ for degree-40 graphs yield better performance than the corresponding SDP correlations. The quality of the SDP correlations is given by the average approximation ratio $\Bar{C}_{\text{app}}^{\text{GW}}$ when applying the Goemans-Williamson (GW) rounding procedure 1024 times.

\subsection{Comparison to Nearest Neighbors}
\label{sec:comparison_to_nearest_neighbors}

In this section, we compare the performance of the CA considering NN and NNN interactions. \Cref{fig:NNN_vs_NN} shows the percentage of optimal solutions found when guided by correlations derived from QAOA at different depths. The experiments are performed on ten graph instances with regular degree 27 and $n=28$ nodes, using 50 mutations with $50n$ iterations each, and the pre-optimization follows the same procedure as in the previous section. 

\begin{figure*}[htbp]
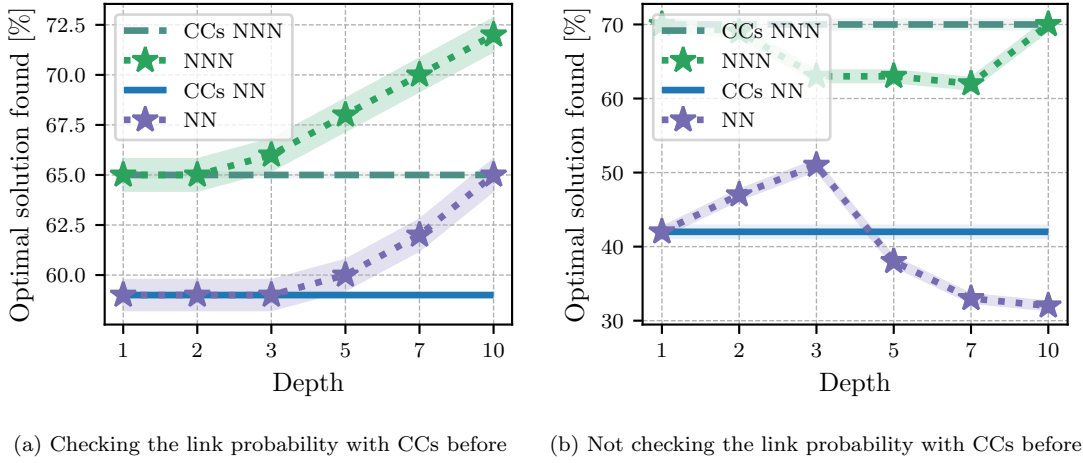

    \centering
    \subfloat[Checking the link probability with CCs before]{%
    \input{NNN_vs_NN_check.pgf}%
        \label{fig:NNN_vs_NN_check}%
    }
    \subfloat[Not checking the link probability with CCs before]{%
    \input{NNN_vs_NN_no_check.pgf}%
        \label{fig:NNN_vs_NN_no_check}%
    }
    \caption{The percentage of optimal solutions found is shown for the CA with nearest neighbor (NN) and next-nearest neighbor (NNN) interactions, guided by QAOA correlations at different depths and coupling constants (CCs). The graphs have regular degree 27 with $n=28$ nodes; 50 mutations with $50n$ iterations are performed. In (a), the link probability $\Tilde{p}_{\text{link}}$ using CCs is evaluated at each iteration when adding a node to the cluster: updates with $\Tilde{p}_{\text{link}}<0$ are rejected, updates with $\Tilde{p}_{\text{link}}=0$ are accepted, and correlations are considered only otherwise. In (b), this evaluation using CCs is omitted.}
    \label{fig:NNN_vs_NN}
\end{figure*}

In (a), the link probability $\Tilde{p}_{\text{link}}$ from \Cref{eq:p_link_nnn} using coupling constants (i.e., $Z_{ij}=J_{ij}$) is evaluated first at each iteration when adding a node to the cluster, even in the correlation-guided case. Note that for this specific evaluation, the maximum function from \Cref{eq:p_link_nnn} is not applied such that the link probability take values smaller than zero. If it is smaller than zero, the update is rejected, while updates with $\Tilde{p}_{\text{link}}=0$ are always accepted; only when it is larger than zero are the correlations taken into account. In this setting, the performance improves monotonically with increasing QAOA depth. In (b), this preliminary check is omitted, and the performance can decrease with increasing depth. This behavior can be attributed to degeneracies that mix correlations of different ground states and thereby reduce their reliability. Including coupling constants incorporates additional information, which mitigates the effect of such unreliable correlations. 

Consequently, we use the additional coupling-constant-based evaluation of $\Tilde{p}_{\text{link}}$ throughout the remainder of this work. Furthermore, the CA considering NNN interactions outperforms the CA considering nearest-neighbor interactions for all data points. Finally, note that SA finds the solution in $28.3\,\%$ of the cases.

We now analyze the percentage point differences in the number of optimal solutions found for graphs with $n=28$ and varying degrees, using the CA considering both NN and NNN interactions, guided by QAOA correlations at depth three and coupling constants, and compare the results to SA in \Cref{fig:NNN_vs_NN_deg}. For each data point, 50 mutations with $50n$ iterations are performed, using the same parameter grid as in the previous section, with the best-performing configuration selected in each case. 

The comparison is structured as follows: the left panel shows the percentage point differences between the QAOA-guided CA and the CA guided by coupling constants, the middle panel displays the individual performance differences for both methods, and the right panel compares the CA guided by coupling constants with SA. As the degree increases, and thus the level of frustration, improvements tend to grow, although this trend is not strictly monotonic, particularly for NN interactions in the left panel and in the comparison with SA. This behavior is attributed to small-graph effects. 

Across almost all instances, the CA considering NNN interactions performs at least as well as the version considering NN interactions, with the only exception being the blue curve at degree ten. We attribute that again to small-graph effects. Moreover, SA is consistently outperformed, even by the CA guided solely by coupling constants.

\begin{figure*}[htbp]
    \centering
    \input{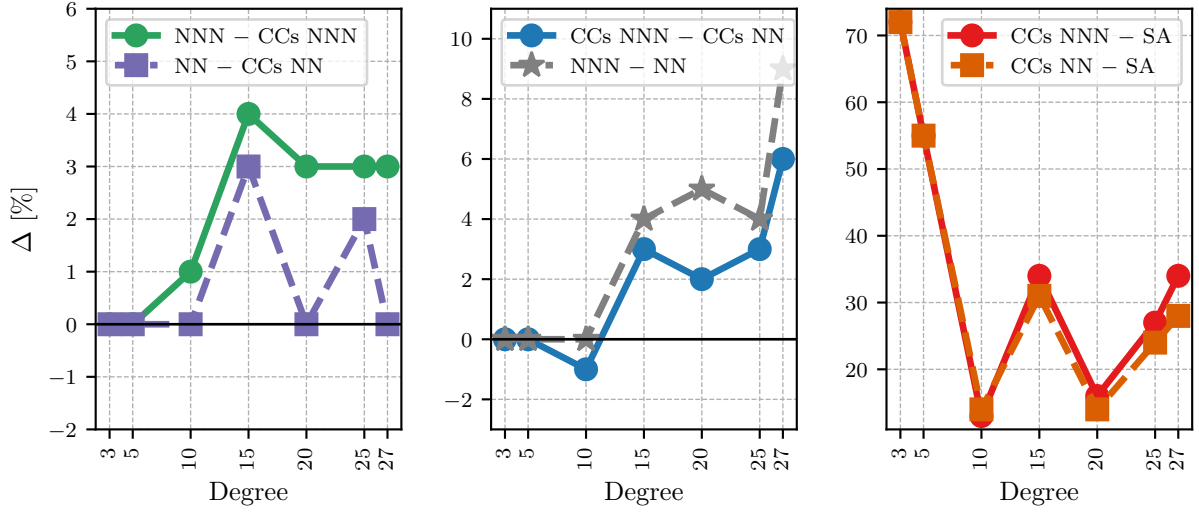}
    \caption{The percentage point differences in optimal solutions found are shown for $n=28$ for the CA considering nearest neighbors (NNs) and next-nearest neighbors (NNNs), using QAOA correlations at depth three, coupling constants (CCs), and for simulated annealing (SA). Each data point is obtained from 50 mutations with $50n$ iterations. The left panel presents the difference between the QAOA-guided CAs and the CC-guided CAs, the middle panel shows the individual performance differences for both approaches, and the right panel compares the CC-guided CAs with SA.}
    \label{fig:NNN_vs_NN_deg}
\end{figure*}

Finally, in \Cref{fig:NN_NNN_sweet_spots}, we investigate the performance of the CA considering NNN interactions on larger graphs of size $n=100$ using thermal samples and examine whether a regime exists in which the samples themselves do not yet contain an optimal solution, while the CA is nevertheless able to find it. To this end, we evaluate the time-to-solution (TTS), a standard performance metric for stochastic algorithms that accounts for both the runtime of a single run and the probability of success. Let $p_{\mathrm{succ}}$ denote the empirical probability that a single run finds a solution of the desired quality, and let $\bar{t}_{\mathrm{run}}$ be the average runtime per run. The number of repetitions required to achieve a success probability of at least $r \in (0,1)$ is then given by
\begin{equation}
N_r = \frac{\ln(1 - r)}{\ln(1 - p_{\mathrm{succ}})},
\end{equation}
and the corresponding time-to-solution is defined as
\begin{equation}
\mathrm{TTS}_r = \bar{t}_{\mathrm{run}} N_r.
\end{equation}
The resulting TTS values for $r = 99\,\%$ are shown in \Cref{fig:NNN_vs_NN_deg_MC_TTS} for degrees three, 20, and 40, together with the bar diagrams in \Cref{fig:samples_opts_found}, which indicate for each of the ten graph instances whether the samples contain at least one optimal solution. For example, for degree-20 graphs at $\beta_s = 0.4$, the samples contain an optimal solution for three instances.

\begin{figure*}[htbp]
    \centering
    \subfloat[$\text{TTS}_{0.99}$ for the CA guided by CCs, thermal correlations, and for SA]{%
    \input{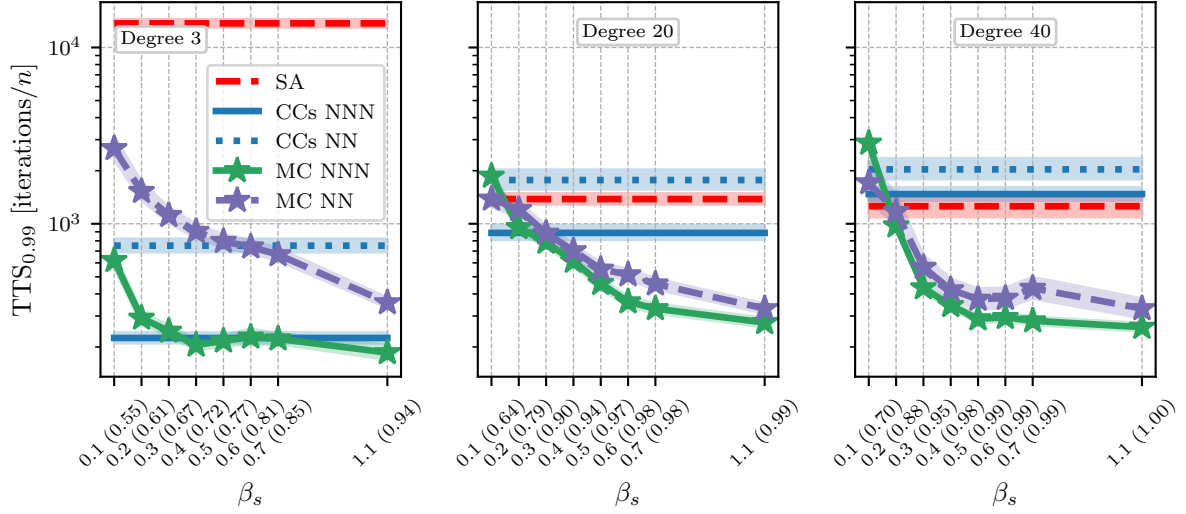}%
        \label{fig:NNN_vs_NN_deg_MC_TTS}%
    }\par
    \subfloat[The number of optimal solutions present in the thermal samples]{%
    \input{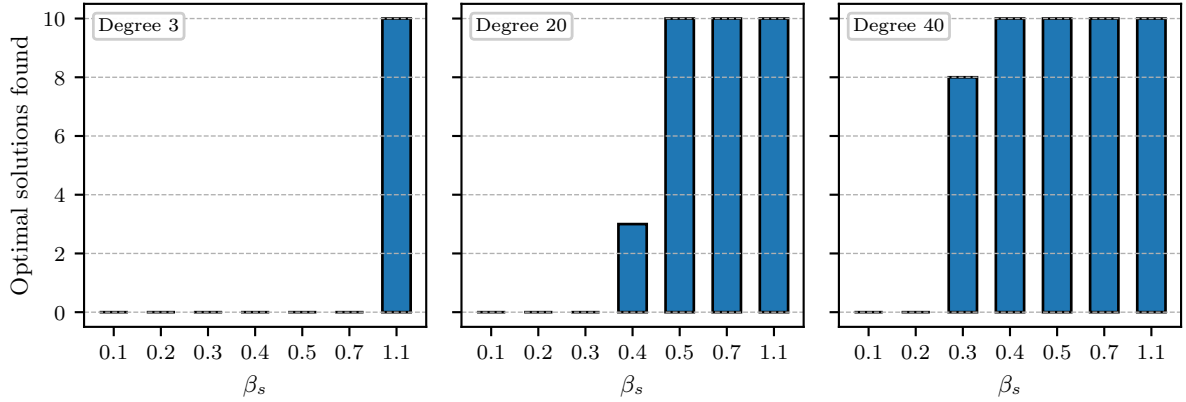}%
        \label{fig:samples_opts_found}%
    }
    \caption{In (a), the times to solution (TTS) with successful run probability $r = 99\,\%$ are shown for the CA considering next-nearest neighbor (NNN) interactions guided by thermal correlations (MC) on graphs with $n = 100$ and degrees three, 20, and 40, together with the coupling-constant (CC) version of the CA and simulated annealing (SA). The values in brackets on the x-axis denote the mean approximation ratios of the samples used to compute the correlations guiding the CA. In (b), bar diagrams are shown that indicate, for each of the ten graph instances and inverse sample temperatures $\beta_s$, whether the thermal samples contain an optimal solution.}
    \label{fig:NN_NNN_sweet_spots}
\end{figure*}

As observed previously, for degree-three graphs the improvement of the CA employing thermal samples over SA is substantial, whereas achieving noticeable gains over the CA guided by coupling constants remains challenging. In particular, only a slight improvement, more pronounced for nearest-neighbor interactions, is observed once all optimal solutions are already present in the samples (i.e., $\beta_s=1.1$). In contrast, for higher degrees, SA begins to outperform the coupling-constant-guided CA, while improved samples become increasingly beneficial. For degree-20 graphs, slight improvements over the coupling-constant version appear already when not all optimal solutions are contained in the samples (i.e., $\beta_s \in \{0.3,0.4\}$), whereas for degree-40 graphs these improvements occur at lower values of $\beta_s \in \{0.2,0.3\}$. 

Overall, these results indicate that the CA performs well at low degrees when relying solely on graph topology, as encoded by the coupling constants, whereas for higher degrees more informative correlations become increasingly important and exhibit a regime of particularly strong performance. This behavior can be attributed to the increased frustration in higher-degree graphs.

Finally, lines of constant solution quality are shown in Appendix~\ref{apx:pareto_analysis}, where the QAOA depth is plotted against the number of iterations. These results further highlight the role of frustration in the QGCA, showing that for higher degrees (i.e., increased frustration), improved correlations lead to stronger performance gains.

\section{Non-Degenerate Chook Problems}
\label{sec:non_degenerate_chook_problems}

In this section, we study a class of Ising spin-glass instances with \emph{planted} ground states generated via \emph{tile planting} from the Chook library, where the full instance is constructed from local subproblems (tiles) with known optimal configurations that are combined consistently. We focus on square-lattice tile-planted problems with probability parameters $p_1=1$ and $p_2=p_3=0$ introduced in~\cite{Perera2020}, as this choice yields some of the most challenging instances within the Chook library for this topology~\cite{Perera2020-Hardness}. A defining property of this ensemble is that the ground state is ferromagnetic and non-degenerate up to the global $\mathbb{Z}_2$ symmetry, which, as shown below, leads to particularly strong performance of the CA.

\subsection{Preliminary Experiments}

The effectiveness of the CA, guided by coupling constants (considering NN and NNN interactions), SDP correlations (considering NNs), and random clusters (RCs), is shown in \Cref{fig:chook_normal} for ten graph instances of size $n=100$, together with results obtained using SA. Each data point is averaged over 50 repetitions per instance. For the coupling-constant and RC variants, we use $\ell_{\text{scale}}=1.0$ and for the SDP-guided algorithm we use $\ell_{\text{scale}}=7.5$, which yield the best performance. For the NNN version we use $e=0.7$. All CA variants, including RCs, outperform SA, while the version considering NNN interactions and coupling constants converges faster than the nearest-neighbor version. In particular, the SDP-guided CA clearly outperforms all other approaches and identifies all optimal solutions after only $11n$ iterations.

As shown in \Cref{fig:chook_scaling_SDP}, applying the Goemans-Williamson (GW) rounding procedure to the SDP correlations such that the time required for hyperplane rounding matches the solving time of the CA with SDP correlations yields at least one optimal solution for seven of the ten instances. However, as will be demonstrated later in \Cref{sec:scaling_analysis}, the CA guided by SDP correlations remains highly effective for larger graph instances, where the rounded solutions no longer reach optimality.

\begin{figure*}[htbp]
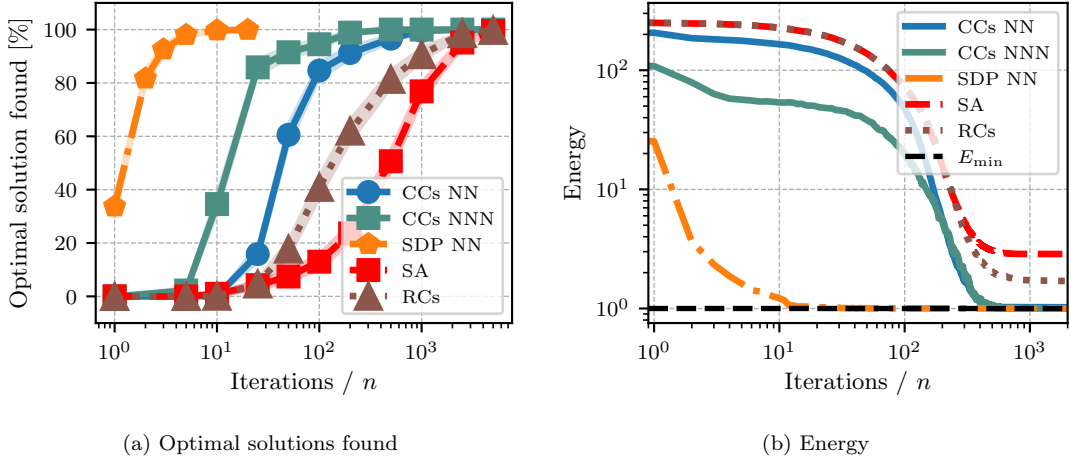

    \centering
    \subfloat[Optimal solutions found]{%
    \input{Chook_C_1_Normal.pgf}%
        \label{fig:Chook_C_1_Normal}%
    }
    \subfloat[Energy]{%
    \input{Chook_C_1_Normal_Energies.pgf}%
        \label{fig:Chook_C_1_Normal_Energies}%
    }
    \caption{In (a), the percentage of optimal solutions found is shown, while in (b) the absolute energy is plotted as a function of the number of iterations of the CA and simulated annealing (SA) for ten tile-planted graphs of size $n=100$. The CA is guided either by coupling constants (CCs), considering nearest-neighbor (NN) or next-nearest-neighbor (NNN) interactions, by correlations obtained from semidefinite programming (SDP) restricted to nearest-neighbor interactions, or by randomly generated clusters (RCs).}
    \label{fig:chook_normal}
\end{figure*}

\subsection{Thermal Correlations}
\label{sec:thermal_correlations}

Next, we analyze the performance of the CA on tile-planted graphs when guided by thermal correlations, in order to better understand its strong performance on spin glasses with unique ground states. In \Cref{fig:Chook_opts_C_1_MC}, the number of optimal solutions found is shown as a function of the number of iterations for the CA guided by thermal correlations obtained using the Metropolis-Hastings algorithm at different inverse sampling temperatures $\beta_s$. For comparison, results for the CA guided by coupling constants (considering NNs) and by RCs are also included. 

For the coupling-constant and RC variants, we use $\ell_{\text{scale}}=1.0$, which yields the best performance, while for the thermal-correlation-guided CA we consider $\ell_{\text{scale}}\rightarrow\infty$ (left) and $\ell_{\text{scale}}=2.0$ (right), referring to the former as the greedy variant. The quality of the samples is analyzed in Appendix~\ref{apx:analysis_of_thermal_samples_for_chook}. For one instance, an optimal solution is found for $\beta_s \ge 0.45$, and for another instance for $\beta_s \ge 0.5$, whereas for the remaining instances the samples do not contain any optimal solutions.

\begin{figure*}[htbp]
    \centering
    \subfloat[Performance of the CA]{%
    \input{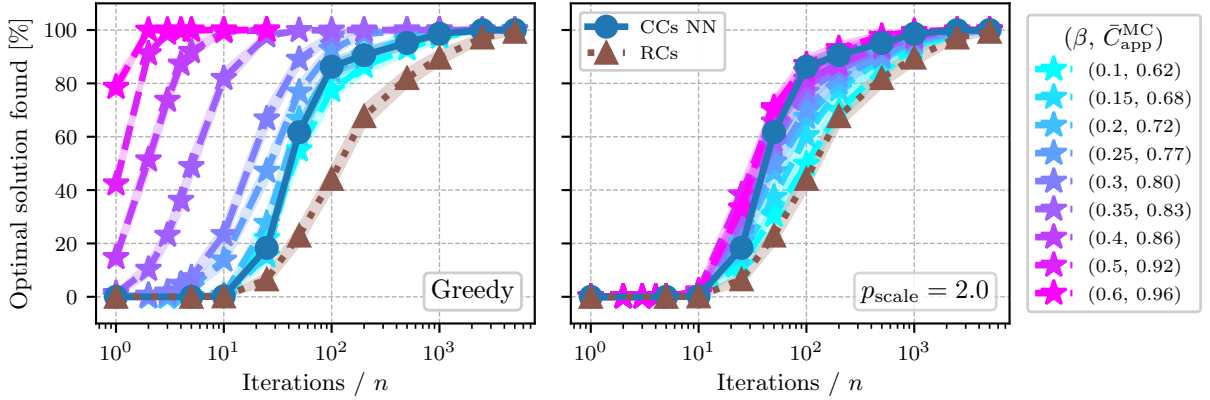}%
        \label{fig:Chook_opts_C_1_MC}%
    }\par
    \subfloat[Distribution of the two-point correlation values]{%
    \input{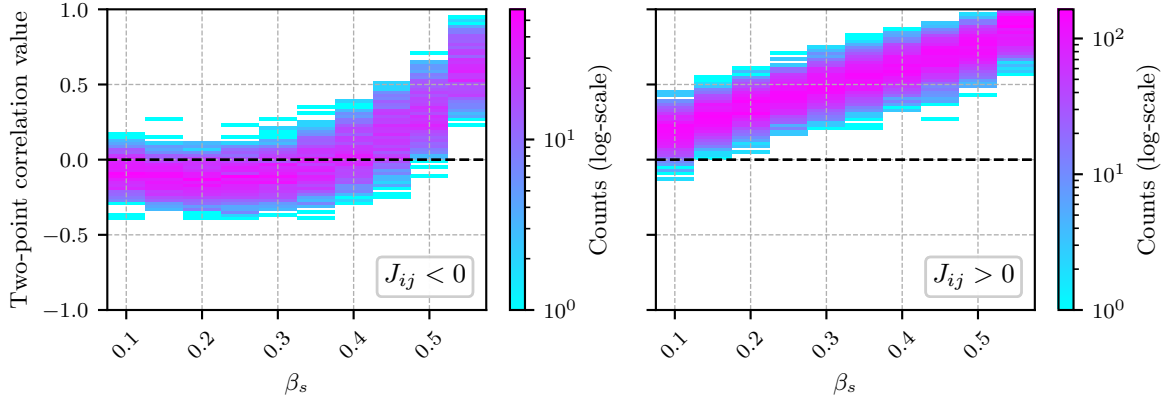}%
        \label{fig:Chook_C_1_frustration}%
    }
    \caption{The percentage of optimal solutions found by the CA employing thermal correlations sampled using the Metropolis-Hastings (MH) algorithm is shown in (a) for $\ell_{\text{scale}}\rightarrow\infty$ (i.e., greedy) and $\ell_{\text{scale}}=2.0$, together with the coupling-constants (CCs) version of the algorithm, both considering only nearest neighbors (NNs). The CA using random clusters (RCs) is also included. The graphs are the same as those used in \Cref{fig:chook_normal}. In (b), the $Z_{ij}^{\text{MC}}$ values computed from Metropolis-Hastings sampling and grouped into bins are shown for edges with $J_{ij}<0$ (left) and $J_{ij}>0$ (right), for the same inverse temperatures as in (a).}
    \label{fig:Chook_C_1_MC}
\end{figure*}

We observe that the CA guided by thermal correlations performs significantly better in the greedy setting and already matches the performance of the coupling-constants variant at the highest sampling temperature considered, i.e., $(\beta_s,\bar{C}_{\text{app}}^{\text{MC}})=(0.1,0.62)$. Notably, for $(\beta_s,\bar{C}_{\text{app}}^{\text{MC}})=(0.6,0.96)$, all optimal solutions are found after only $2n$ iterations, despite the samples containing only two optimal solutions. 

This behavior can be understood from \Cref{fig:Chook_C_1_frustration}, where the two-point correlations $Z_{ij}^{\text{MC}}$ are binned separately for edges with $J_{ij}<0$ (left) and $J_{ij}>0$ (right). Since the ground state is unique and ferromagnetic, all correlations approach one in the limit $\beta_s \to \infty$. At the same time, correlations corresponding to $J_{ij}<0$ remain predominantly negative at high sampling temperatures and begin to increase around $\beta_s \approx 0.3$, which coincides with the point where the performance of the CA improves sharply as the sampling temperature is lowered. Around $\beta_s \approx 0.6$, the correlations are essentially correct, and in the greedy variant, where they are always accepted, the CA converges rapidly. In this regime, the system effectively behaves like a graph without frustrated spins. 

Importantly, this effect is not limited to spin glasses with ferromagnetic ground states but rather applies more generally to systems with non-degenerate ground states. A related result has been presented in~\cite{Eder2026} for the dial-a-ride problem, proving that the CA converges in expected $\mathcal{O}(n^3)$ steps when guided by optimal correlations. While access to optimal correlations would render post-processing unnecessary, this result nevertheless illustrates the potential strength of the method.

\subsection{Scaling Analysis}
\label{sec:scaling_analysis}

Finally, we analyze the scaling behavior of the CA on the Chook instances for increasing graph sizes. To this end, \Cref{fig:Chook_scaling} shows the TTS for $r=99\,\%$, measured both in CPU time (a) and in number of iterations (b). For each graph size $n$, ten Chook instances are considered. The parameter $\ell_{\text{scale}}$ is optimized over selected values between $1$ and $1000$, including $\infty$, and is ultimately set to $1.0$, $\infty$, $6.0$, and $1.0$ for coupling constants with NNs and NNNs, SDP, and RCs, respectively. For the NNN version, we choose $e=0.7$. 

For each $n$, the number of iterations is optimized with respect to the TTS over the set $\{2^l n \mid l\in\{5,6,\dots,14\}\}$. Each data point is based on 50 mutations per graph instance. The implementation is carried out in Julia~\cite{Julia2017}.

\begin{figure*}[htbp]
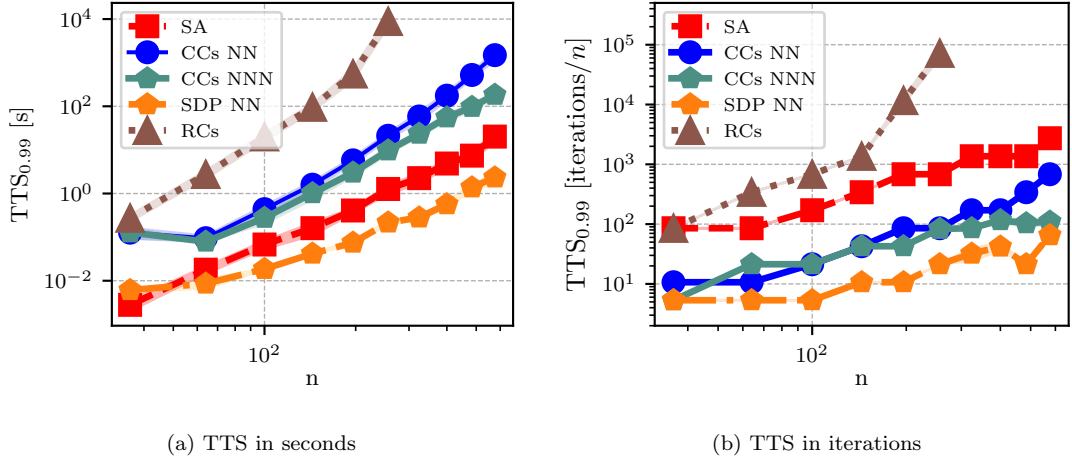

    \centering
    \subfloat[TTS in seconds]{%
    \input{Chook_scaling_TTS.pgf}%
        \label{fig:Chook_scaling_TTS}%
    }
    \subfloat[TTS in iterations]{%
    \input{Chook_scaling_TTS_sweeps.pgf}%
        \label{fig:Chook_scaling_TTS_sweeps}%
    }
    \caption{The scaling of the time-to-solution (TTS) for $r=99\,\%$ as a function of the graph size $n$ is shown in (a) CPU time and (b) number of iterations for the CA guided by coupling constants (CCs), considering nearest neighbors (NNs) and next-nearest neighbors (NNNs), SDP correlations restricted to nearest neighbors, and RCs, together with SA. In terms of CPU time, only the SDP-guided CA outperforms SA, whereas the coupling-constants variants also outperform SA when considering the scaling in the number of iterations.}
    \label{fig:Chook_scaling}
\end{figure*}

From the CPU-time scaling shown in (a), we observe that only the SDP-guided variant of the CA outperforms SA. In addition, the coupling constants with NNNs variant exhibits better scaling than the coupling constants with NN version, while the RC-guided CA shows the poorest scaling behavior. In contrast, the iteration-based scaling shown in (b), where implementation-specific overheads are negligible and only algorithmic scaling is considered, reveals that SA is also outperformed by the coupling-constants variants of the CA. 

The TTS in CPU time is fitted with the form $a n^b$, yielding exponents $b=3.13$ (SA), $4.4$ (coupling constants NN), $3.63$ (coupling constants NNN), $2.56$ (SDP NN), and $4.68$ (RCs). The first data point is excluded from the fits for both coupling-constants variants, and the final data point for RCs is omitted, as the maximum number of iterations was insufficient to properly optimize the TTS. This analysis highlights the potential strength of the presented CA; however, the time required to compute the SDP correlations is not included, as it significantly exceeds the runtime of the CA itself.

Finally, the quality of the SDP samples is analyzed in \Cref{fig:chook_scaling_SDP} by applying the GW rounding procedure such that the time required for hyperplane rounding matches the solving time of the CA with SDP correlations. In (a), the mean approximation ratios $\Bar{C}_{\text{app}}^{\text{GW}}$ are shown together with standard deviations and min-max ranges. In (b), the number of instances for which the rounded samples contain an optimal solution is reported. For example, for $n=100$, optimal solutions are obtained for seven out of ten instances. For $n\geq196$, no optimal solutions are found in the rounded samples, whereas the performance of the CA remains stable for larger graph sizes.

\begin{figure*}[htbp]
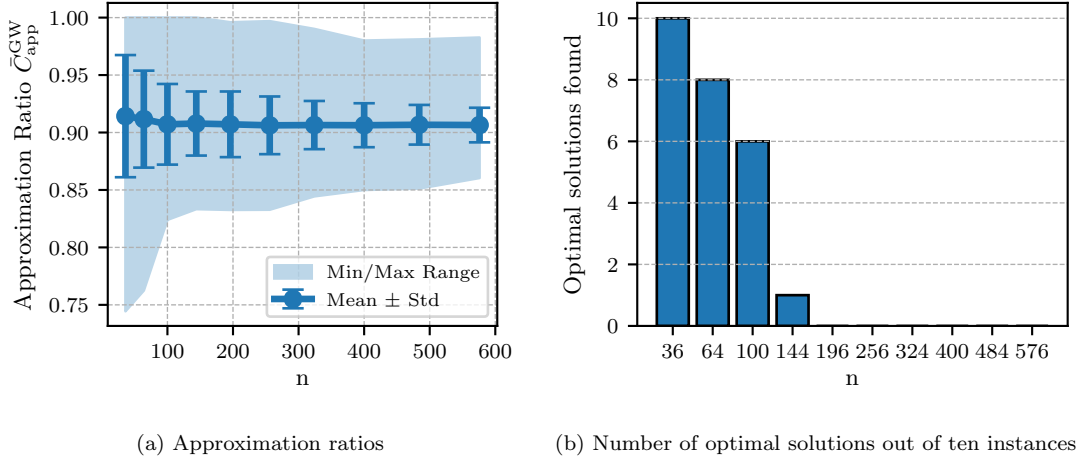

    \centering
    \subfloat[Approximation ratios]{%
    \input{Chook_scaling_SDP_capp_TTSbudget.pgf}%
        \label{fig:Chook_scaling_SDP_capp}%
    }
    \subfloat[Number of optimal solutions out of ten instances]{%
    \input{Chook_scaling_SDP_nums_opt_found_TTSbudget.pgf}%
        \label{fig:Chook_scaling_SDP_nums_opt_found}%
    }
    \caption{In this figure, the quality of the semidefinite-programming (SDP) samples is evaluated by applying the Goemans-Williamson (GW) rounding procedure such that the time required for hyperplane rounding matches the solving time of the CA with SDP correlations. Panel (a) shows the mean approximation ratios $\bar{C}_{\text{app}}^{\text{GW}}$ together with their standard deviations and min–max ranges. Panel (b) reports how often the rounded samples contain an optimal solution. For example, for $n=100$, optimal solutions are obtained for seven out of ten instances, whereas for $n \ge 196$ no optimal solutions are found.}
    \label{fig:chook_scaling_SDP}
\end{figure*}

\section{Quantum-Informed Markov-Chain Monte Carlo}
\label{sec:quantum_informed_markov_chain_monte_carlo}

Finally, as an outlook, we present a further extension of the CA by formulating it as a valid MCMC algorithm whose stationary distribution coincides with a target distribution~$\pi$, i.e., a stochastic sampling procedure that generates samples asymptotically distributed according to~$\pi$. In particular, we consider the Boltzmann distribution, $\pi(s) \propto \exp(-\beta E(s))$, thereby substantially extending the applicability of the informed CA. While the computation of correlations used to guide the CA may already identify optimal solutions in an optimization setting, this does not address the central goal of MCMC methods, namely sampling from the target distribution. Consequently, this extension is relevant for both NISQ-era and fault-tolerant settings, where combining tailored correlation structures with cluster updates can significantly enhance sampling performance.

To ensure validity, the MCMC algorithm must satisfy ergodicity (i.e., any state can be reached from any other via a sequence of updates) and detailed balance (transition probabilities reproduce the target distribution $\pi$ in equilibrium) with respect to~$\pi$. To address ergodicity, we modify the link probability in \Cref{eq:p_link} to
\begin{align}
    &p_{\text{link,MCMC}}^{(i,j)}
    = \nonumber \\ &\max\Bigg[0,\ 
        1- \exp \Bigg(-\frac{\ell_{\text{scale}}}{\ell_{\text{perc}}} 
        x_i x_j Z_{ij} \Bigg)
    \Bigg],
\label{eq:mcmc_p_link}
\end{align}
which guarantees a strictly non-zero probability of flipping a single spin, as the link probability remains strictly smaller than one. This ensures that the algorithm does not get trapped in a subset of configurations and can explore the full configuration space. Furthermore, we restrict the construction to NNs only, enabling an efficient implementation that satisfies detailed balance, which is explained in the following. The acceptance probability of a proposed move is then chosen according to the Metropolis-Hastings update rule
\begin{equation}
    p_{\text{acc}}(x \to x')
    \;=\;
    \min\!\left(
        1,\;
        \frac{\pi(x')\, q(x \mid x')}{\pi(x)\, q(x' \mid x)}
    \right),
    \label{eq:bac_detailed_balance_criterion}
\end{equation}
ensuring that transitions between states occur with the correct relative probabilities.

To evaluate the acceptance probability, the proposal probability ratio \( q(x' \mid x) / q(x \mid x') \) must be computed efficiently. This is achieved by modifying the cluster construction such that no shrinking step is performed when adding nodes; instead, nodes are added while leaving all other edges unchanged. Under this construction, the probabilities associated with adding cluster nodes and rejecting boundary nodes factorize, with the former canceling exactly in the proposal probability ratio, while the latter can be computed efficiently. As a result, the order of cluster construction does not affect the acceptance probability. The following theorem formalizes this result.

\begin{theorem}
Let $G=(V,\mathcal{E})$ be an undirected graph with $n$ nodes, and let $C_{x,x'}\subseteq V$ denote the set of vertices whose simultaneous flip maps a configuration $x\in\{-1,1\}^n$ to $x'$. Define
\[
\mathcal{E}_{x,x'} :=
\bigcup_{\substack{
v_i \in C_{x,x'} \\
v_j \notin C_{x,x'}
}}
\{v_i,v_j\}
\subseteq \mathcal{E}
\]
as the set of edges crossing the boundary of the cluster $C_{x,x'}$.  
The proposal distribution of \Cref{alg:mcmc}, restricted to NN interactions and without a shrinking step (i.e., when a node is added to the cluster, all edges remain untouched except those fully contained in the cluster), satisfies
\[
\frac{q(x'\mid x)}{q(x\mid x')}
=
\prod_{\{i,j\}\in \mathcal{E}_{x,x'}}
\frac{1-p^{(i,j)}_{\mathrm{link,MCMC}}(x)}
     {1-p^{(i,j)}_{\mathrm{link,MCMC}}(x')}.
\]
\label{thm:detailed_balance_mcmc}
\end{theorem}

The proof is deferred to Appendix~\ref{apx:efficient_calculation_of_the_proposal_ratios_for_the_quantum_informed_mcmc_algorithm}. As a consequence of the factorization over boundary nodes, the acceptance probability can be evaluated efficiently. The corresponding cluster-building procedure is denoted by the function \texttt{CreateMcmcClusterNN} and is used in \Cref{alg:mcmc}, which provides the pseudocode of the quantum-informed MCMC algorithm based on cluster construction.

\begin{algorithm}[htpb]
\caption{Quantum-Informed MCMC.}
\label{alg:mcmc}
\SetAlgoLined
\KwIn{Hamiltonian $H$, correlation matrix $Z$, inverse temperature $\beta$, total sweeps $T$, link parameter $\ell_{\text{scale}}$.}
\KwOut{Sample list $\mathcal{S}$.}

Initialize $x$ uniformly at random;\\
Initialize empty sample list $\mathcal{S}$;\\
$t\leftarrow0$;\\
\While{$t<T$}{
$i\leftarrow \mathrm{Uniform}({1,n})$;\\
$\mathcal{C} \gets \texttt{CreateMcmcClusterNN}(i,x,H,Z,\ell_{\text{scale}})$;\\
$x' \gets \mathrm{Flip}(x, \mathcal{C})$;\\
$\Delta E \leftarrow H(x') - H(x)$;\\
$A \leftarrow \min\left(1, e^{-\beta \Delta E}\frac{q(x'\mid x)}{q(x \mid x')}\right)$;\\
Draw $u\sim \mathrm{Uniform}(0,1)$;\\
\If{$u < A$}{$x\leftarrow x'$;}
Append $x$ to $\mathcal{S}$;\\
$t\leftarrow t+|\mathcal{C}|$;
}
\Return $\mathcal{S}$;
\end{algorithm}

In the following, we present numerical results obtained with the quantum-informed MCMC algorithm. In addition to the steps given in \Cref{alg:mcmc}, we perform a single-spin flip with probability $35\,\%$ and a global spin flip (i.e., flipping all spins) with probability $5\,\%$, while cluster updates are applied otherwise. These additional updates improve ergodicity and enable efficient exploration of both sectors of the $\mathbb{Z}_2$-symmetric \textsc{Max-Cut} instances. As before, a single cluster flip is counted as equivalent to $|\mathcal{C}|$ single-spin or global spin flips.

To assess sampling performance, it is important to quantify correlations between successive samples of the Markov chain. Even after the chain has reached equilibrium, consecutive configurations can remain correlated, which reduces the number of effectively independent samples. This behavior is characterized by the normalized autocorrelation function of an observable $f$, defined as
\begin{equation}
a_f(l) :=
\frac{
\mathbb{E}\!\left[(f(s_k) - \mathbb{E}[f])(f(s_{k+l}) - \mathbb{E}[f])\right]
}{
\mathbb{E}\!\left[(f(s_k) - \mathbb{E}[f])^2\right]
},
\end{equation}
where $l$ denotes the lag, i.e., the number of steps separating the two samples. A rapid decay of $a_f(l)$ indicates that the Markov chain quickly loses memory of previous states and therefore mixes efficiently.

In \Cref{fig:mcmc_ACF}, we show the autocorrelation function $a$ as a function of the lag $l$. We compare the cluster-based MCMC algorithm guided by correlations derived from QAOA at depths $p=1$ and $p=3$ to baseline MCMC algorithms employing only single-spin flips (SSFs) or uniform spin flips (USFs), i.e., updates in which a subset of spins is selected uniformly at random and flipped. In both baseline variants, global spin flips are additionally performed with a probability of $5\,\%$. To comprehensively assess sampling performance, we analyze both the energy and the magnetization. Since global spin flips cause the autocorrelation function of the magnetization to decay immediately, we instead consider the autocorrelation of the absolute magnetization, which provides a more meaningful measure of mixing behavior.

\begin{figure*}[htbp]
    \centering
    \subfloat[Degree 3, absolute magnetization]{%
    \input{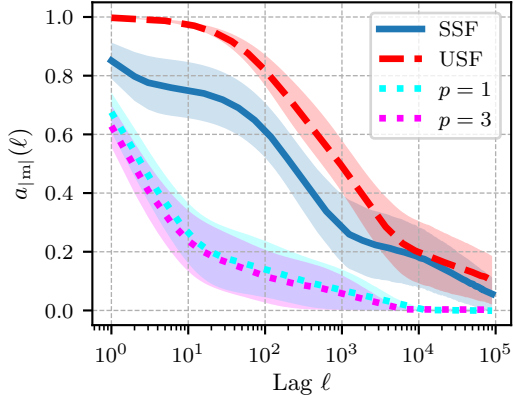}%
        \label{fig:mcmc_deg_3_type_MAG}%
    }
    \subfloat[Degree 3, energy]{%
    \input{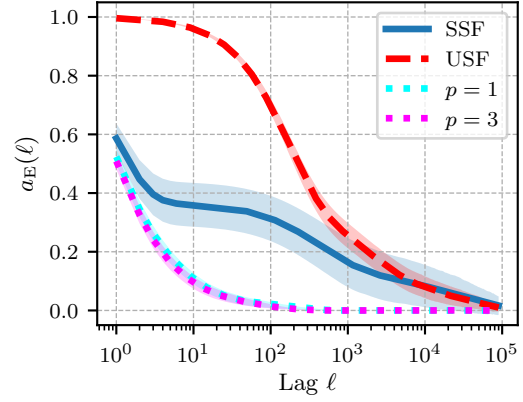}%
        \label{fig:mcmc_deg_3_type_ENE}%
    }
    \par\medskip
    \subfloat[Degree 27, absolute magnetization]{%
    \input{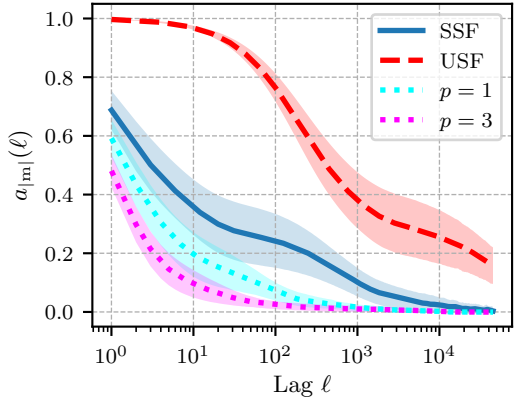}%
        \label{fig:mcmc_deg_27_type_MAG}%
    }
    \subfloat[Degree 27, energy]{%
    \input{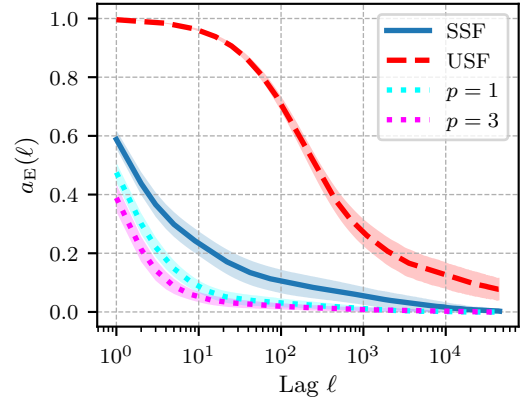}%
        \label{fig:mcmc_deg_27_type_ENE}%
    }
    \caption{The autocorrelation functions are shown as a function of the lag for different MCMC variants. In all panels, results for the CA guided by correlations derived from QAOA at depths $p=1$ and $p=3$ are presented alongside the single-spin-flip (SSF) and uniform-spin-flip (USF) variants. Panels (a) and (b) display the autocorrelation of the absolute magnetization $|m|$ and the energy $E$, respectively, for three-regular graphs, while panels (c) and (d) show the corresponding quantities for 27-regular graphs.}
    \label{fig:mcmc_ACF}
\end{figure*}

In all experiments, the scaling factor is set to $\ell_{\text{scale}} = 1$. For each algorithmic variant, ten random graph instances of size $n=28$ are generated and each instance is evaluated ten times. We consider both three-regular and 27-regular graphs. The inverse temperature is chosen as $\beta = 3$ for three-regular graphs and $\beta = 1$ for 27-regular graphs. To reduce the influence of the initial configuration, the first $10{,}000$ samples are discarded as burn-in. After this equilibration phase, $3 \times 10^6$ MC sweeps (corresponding to $3 \times 10^6 n$ update steps) are performed for each run.

For both graph classes, the quantum-informed cluster MCMC algorithm exhibits a significantly faster decay of correlations compared to the single-spin and uniform-spin flip baselines. Increasing the quality of the quantum correlations from QAOA depth $p=1$ to $p=3$ leads to only minor additional improvements for three-regular graphs. In contrast, for higher-degree graphs, improved quantum correlations yield substantially stronger gains. This difference can be attributed to the higher level of frustration and the resulting increased difficulty of sampling in these instances.

\begin{sloppypar}
Additional results, including post-burn-in acceptance probabilities and magnetization histograms of the quantum-informed MCMC algorithm, are provided in Appendix~\ref{apx:details_and_magnetizations_for_the_numerical_results_of_the_quantum_informed_mcmc_algorithm}. For low-degree graphs, the magnetization distributions are nearly symmetric around zero, whereas for degree-$27$ graphs a slight asymmetry is observed. In an ideal MCMC sampler, the magnetization histogram should be perfectly symmetric due to the $\mathbb{Z}_2$ symmetry of the problem and the invariance of the Boltzmann distribution under global spin inversion. The observed asymmetry therefore indicates that certain regions of the configuration space are explored less efficiently, suggesting the presence of metastable subspaces or slow mixing modes that remain challenging even for the cluster-based updates. Further improvements may be achieved by incorporating ideas from generalized Swendsen-Wang-type methods~\cite{Barbu2005,Barbu2020}, which allow the construction of cluster proposals based on problem-dependent information rather than coupling constants while preserving the desired stationary distribution.
\end{sloppypar}

\clearpage

\section{Conclusion and Future Work}
\label{sec:conclusion_and_future_work}

In this work, we investigated correlation-guided CAs for combinatorial optimization and sampling, extending previous approaches~\cite{Eder2025-Cluster,Eder2026} by incorporating NNN information and formulating a corresponding MCMC variant. Our results again show that correlation-informed cluster updates can significantly improve performance over local methods, particularly in highly frustrated problems where simple topological information becomes insufficient.

A key finding is the impact of incorporating NNN information into the cluster construction. Across all considered correlation sources (thermal, coupling constants, QAOA, and SDP), we observed that including NNNs consistently improves the performance of the CA itself, even when guided only by coupling constants. For thermal correlations, we saw that the gains in terms of TTS are comparatively moderate, whereas for non-optimized iteration counts the improvements can be more pronounced. At the same time, the stronger baseline performance of the NNN-enhanced CA implies that higher levels of frustration are required for additional correlation information to yield noticeable improvements over the coupling-constant-guided version. We further demonstrated that the CA performs particularly well on non-degenerate instances, such as the Chook benchmarks, where the absence of competing ground states leads to highly informative correlations and rapid convergence.

Finally, we introduced an MCMC extension of the CA targeting a Boltzmann distribution, thereby extending the method beyond pure optimization. While this significantly broadens its applicability, the mixing behavior and robustness of the resulting algorithm are not yet fully understood and require further investigation.

Future work could therefore focus on improving the MCMC variant. In this context, it is natural to draw inspiration from generalized Swendsen-Wang-type algorithms~\cite{Barbu2005,Barbu2020}. Adapting such constructions to correlation-guided settings remains an open problem. In addition, a systematic scaling analysis of the MCMC variant is required to assess its practical relevance. On the quantum side, evaluating the CAs on larger hardware will be essential. More generally, identifying correlation sources that provide sufficiently accurate guidance at low computational cost remains a key challenge for making the method competitive in practice.

\begin{acknowledgments}
The authors thank Aron Kerschbaumer, Jernej Rudi Fin\v{z}gar, Raimel A. Medina, Martin J. A. Schuetz, Helmut G. Katzgraber, and Christian B. Mendl for insightful discussions. During the preparation of this work the authors used ChatGPT from OpenAI based on the GPT-5 architecture, in order to improve the readability and language of the manuscript. After using this tool, the authors reviewed and edited the content as needed and take full responsibility for the content of the published article.
\end{acknowledgments}

\appendix
\clearpage
\onecolumngrid

\section{Alternative Cost Calculations for Next-Nearest Neighbors}
\label{apx:alternative_cost_calculations_for_nnn}

In this section, we analyze the computational cost of cluster updates when NNN interactions are taken into account, as introduced in \Cref{sec:extension_to_next_nearest_neighbors}. The objective is to relate the cost of flipping a cluster of size \(k\) to that of SA, by interpreting a cluster update as effectively equivalent to \(k\) single-spin updates. To this end, we examine how the computational effort of a cluster flip scales with \(k\).

A cluster update consists of three main components: (i) attempting to grow the cluster by adding spins, (ii) updating the graph and auxiliary quantities when a spin is accepted, and (iii) computing the energy change associated with flipping the final cluster. As an initialization step, the algorithm computes an energy-density-like vector
\begin{equation}
\tilde{\rho}_0(i)=\sum_{j\in\mathcal{N}(i)}J_{ij}x_i^0 x_j^0,\qquad \forall i\in V,
\end{equation}
which is evaluated once at the beginning and can therefore be neglected in the scaling analysis. However, it must be updated whenever a spin is added to the cluster.

We first consider the cost of (i), i.e., evaluating whether candidate spins should be included in the cluster. In the worst case of a fully connected graph, the number of such attempts scales as \(\mathcal{O}(kn)\), since up to \(\mathcal{O}(n)\) spins may be tested for each of the \(k\) accepted spins. This corresponds to a scenario in which, at each step, all remaining spins are rejected before one is added. Each individual evaluation requires only constant time, as the link probability \(\tilde{p}_{\text{link}}\) is computed using a single entry of \(\tilde{\rho}_l(i)\), resulting in an \(\mathcal{O}(1)\) operation.

Next, we analyze the cost of (ii), i.e., adding a spin to the cluster. This requires updating both the effective graph representation and the vector \(\tilde{\rho}\). In a fully connected graph, up to \(\mathcal{O}(n)\) edges are affected by such an update, leading to an \(\mathcal{O}(n)\) cost per accepted spin. For a cluster of size \(k\), this results in a total cost of \(\mathcal{O}(kn)\).

Finally, the cost of (iii), computing the energy difference for the completed cluster, also scales as \(\mathcal{O}(kn)\) in the worst case, since the cluster boundary may involve up to \(\mathcal{O}(kn)\) edges. Combining these contributions, the overall complexity of constructing and flipping a cluster of size \(k\) scales as \(\mathcal{O}(kn)\), up to constant factors.

For comparison, SA performs single-spin updates (\(k=1\)), resulting in an \(\mathcal{O}(n)\) cost per update in a fully connected graph. Therefore, from a scaling perspective, a cluster flip of size \(k\) in the CA is computationally equivalent to performing \(k\) single-spin updates in SA.

\section{Pareto Analysis}
\label{apx:pareto_analysis}

In this section, we perform a Pareto analysis of the CA guided by correlations derived from QAOA. The color scale used in \Cref{fig:apx_pareto_qaoa_analysis} is explained in \Cref{tab:apx_clog_color_scale}.

\begin{table}[h!]
\centering
\begin{tabular}{|c|c|c|}
\hline
\textbf{Percentage Found (\%)} & \textbf{$100 - \text{Percentage}$ (\%)} & \textbf{$\log_{10}(100 - \text{Percentage})$} \\
\hline
100.0 & $\sim$0     & $-\infty$ (clipped) \\
99.9  & 0.1         & $-1$ \\
99.0  & 1           & $0$ \\
90.0  & 10          & $1$ \\
50.0  & 50          & $\approx 1.70$ \\
10.0  & 90          & $\approx 1.95$ \\
0.0   & 100         & $2$ \\
\hline
\end{tabular}
\caption[Table explaining the Pareto analysis for the quantum-guided CA]{Mapping of optimal solution percentages to color values using $\log_{10}(100 - \text{Percentage})$ for \Cref{fig:apx_pareto_qaoa_analysis}.}
\label{tab:apx_clog_color_scale}
\end{table}

For degree $10$, achieving a solution quality of $99.9\%$ requires approximately $125n$ iterations and is largely independent of the QAOA depth. In contrast, for degree $27$, the number of required iterations to reach $99\%$ quality decreases substantially, from about $2000n$ at depth $3$ to roughly $750n$ at depth $10$. This demonstrates that for higher-degree instances, improved correlations have a stronger impact on performance, which can be attributed to the increased level of frustration.

\begin{figure*}[htbp]
    \centering
    \subfloat[Degree 10]{%
        \input{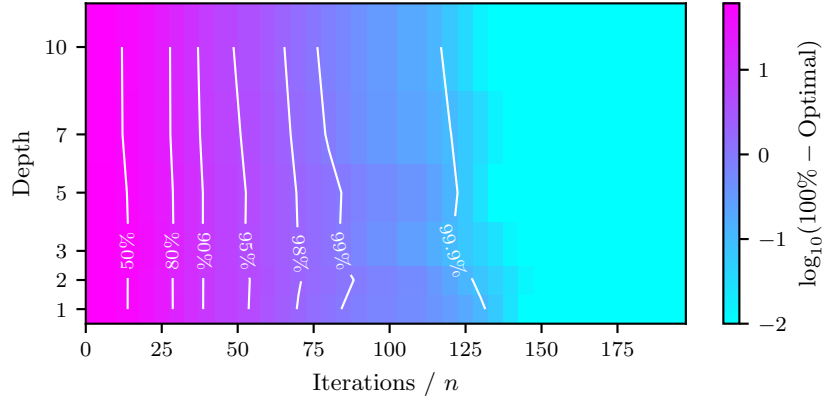}%
        \label{fig:QAOA_Pareto_degree_10}%
    }\\[1ex]
    \subfloat[Degree 20]{%
        \input{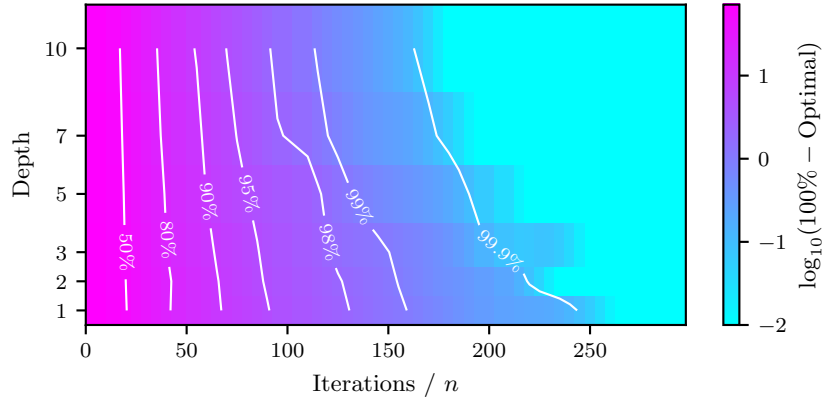}%
        \label{fig:QAOA_Pareto_degree_20}%
    }\\[1ex]
    \subfloat[Degree 27]{
    \input{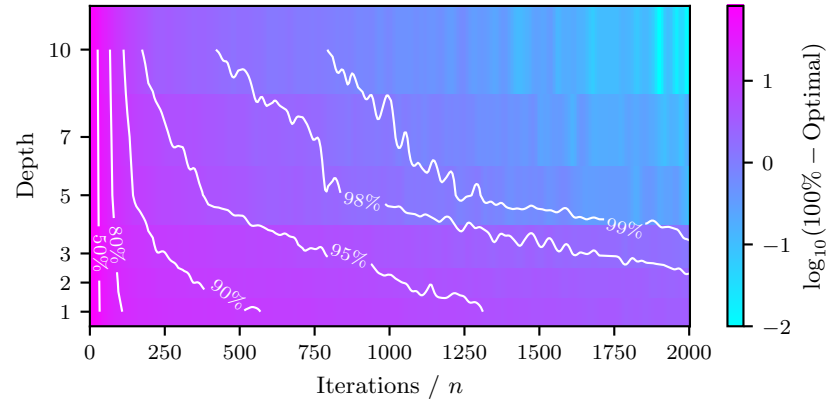}%
        \label{fig:QAOA_Pareto_degree_27}%
    }
    \caption{Lines of constant solution quality (i.e., Pareto analyses) are shown for random regular-degree graphs of size $n=28$. The color scale is explained in \Cref{tab:apx_clog_color_scale}. The CA is run considering next-nearest neighbor (NNN) interactions and is optimized over $e \in \{0.5, 0.6, \dots, 1.0\}$ and $\ell_{\text{scale}} \in \{1.0, 3.0, 6.0, 10.0, 15.0, 20.0, 50.0\}$. Each bin is evaluated 50 times for every parameter combination to ensure statistical consistency.}
    \label{fig:apx_pareto_qaoa_analysis}
\end{figure*}

\clearpage

\section{Analysis of Thermal Samples for Chook}
\label{apx:analysis_of_thermal_samples_for_chook}

\begin{figure*}[h!]
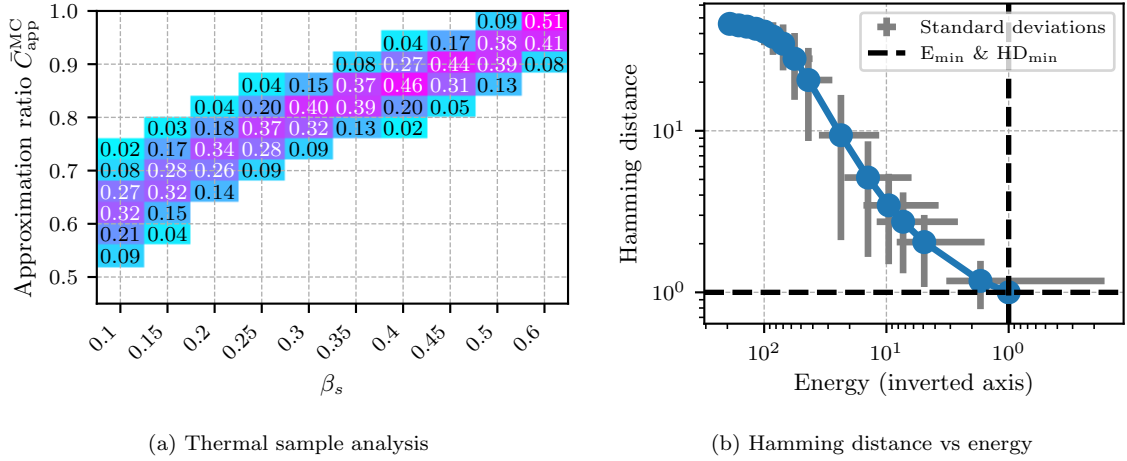

    \centering
    \subfloat[Thermal sample analysis]{%
    \input{Chook_C_1_samples.pgf}%
        \label{fig:Chook_C_1_samples}%
    }
    \subfloat[Hamming distance vs energy]{%
    \input{Chook_C_1_hds_energies.pgf}%
        \label{fig:Chook_C_1_hds_energies}%
    }
    \caption{The thermal samples used in \Cref{sec:non_degenerate_chook_problems} are grouped into bins according to their approximation ratios, and the resulting distributions are shown in (a) with a display threshold of 0.01. The presence of optimal solutions in the samples is discussed in the main text in \Cref{sec:non_degenerate_chook_problems}. In (b), the samples are further analyzed by examining how their Hamming distances to the unique optimal solution (modulo the $\mathbb{Z}_2$ symmetry) relate to the solution quality. While the mean Hamming distance decreases monotonically, the standard deviations exhibit significant overlap, indicating that these instances remain challenging for MC-based methods.}
    \label{fig:Chook_C_1_samples_and_problem}
\end{figure*}

\section{Efficient Calculation of the Proposal Ratios for the Quantum-Informed MCMC Algorithm (\Cref{thm:detailed_balance_mcmc})}
\label{apx:efficient_calculation_of_the_proposal_ratios_for_the_quantum_informed_mcmc_algorithm}

\begin{restatement}
Let $G=(V,\mathcal{E})$ be an undirected graph with $n$ nodes, and let $\mathcal{C}_{x,x'}\subseteq V$ denote the set of vertices whose simultaneous flip maps a configuration $x\in\{-1,1\}^n$ to $x'$. Define
\[
\mathcal{E}_{x,x'} :=
\bigcup_{\substack{
v_i \in \mathcal{C}_{x,x'} \\
v_j \notin \mathcal{C}_{x,x'}
}}
\{v_i,v_j\}
\subseteq \mathcal{E}
\]
as the set of edges crossing the boundary of the cluster $\mathcal{C}_{x,x'}$.  
The proposal distribution of \Cref{alg:mcmc}, restricted to nearest-neighbor (NN) interactions and without a shrinking step (i.e., when a node is added to the cluster, all edges remain untouched except those fully contained in the cluster), satisfies
\[
\frac{q(x'\mid x)}{q(x\mid x')}
=
\prod_{\{i,j\}\in \mathcal{E}_{x,x'}}
\frac{1-p^{(i,j)}_{\mathrm{link,MCMC}}(x)}
     {1-p^{(i,j)}_{\mathrm{link,MCMC}}(x')}.
\]
\end{restatement}

\begin{proof}
Let $P_I(x)$ denote the probability that, for a subgraph $G_I\subseteq G$ containing exactly the vertices in $\mathcal{C}_{x,x'}$, the cluster $\mathcal{C}_{x,x'}$ is constructed.  
Then the proposal probability factorizes as
\[
q(x'\mid x)
=
P_I(x)\,
\prod_{\{i,j\}\in \mathcal{E}_{x,x'}}
\bigl(1-p^{(i,j)}_{\mathrm{link,MCMC}}(x)\bigr).
\]
This follows because, in the absence of a shrinking step, the link probability $p^{(i,j)}_{\text{link,MCMC}}(x)$ is independent of the cluster building order for all edges $\{v_i,v_j\}\in \mathcal{E}_{x,x'}$. Moreover, in order to construct the cluster $\mathcal{C}_{x,x'}$, all such boundary edges must always be rejected.

Finally, the inclusion probabilities satisfy $P_I(x)=P_I(x')$, since $P_I$ consists only of factors of $p^{(i,j)}_{\mathrm{link,MCMC}}$, each depending on the product $x_i x_j$. Because $x$ and $x'$ differ only by a global sign flip on the vertices in $\mathcal{C}_{x,x'}$, we have $x_i x_j = x'_i x'_j$ for all relevant edges. Hence, the $P_I$ factors cancel in the ratio $\frac{q(x'\mid x)}{q(x\mid x')}$, which proves the claim.
\end{proof}

\section{Details and Magnetizations for the Numerical Results of the Quantum-Informed MCMC Algorithm}
\label{apx:details_and_magnetizations_for_the_numerical_results_of_the_quantum_informed_mcmc_algorithm}

\begin{table}[h!]
\centering
\begin{tabular}{lcc}
\toprule
\textbf{Method} & \textbf{Degree 3} & \textbf{Degree 27} \\
\midrule
SSF            & $5.15 \pm 0.02$ & $6.99 \pm 0.66$ \\
USF            & $5.03 \pm 0.02$ & $5.17 \pm 0.07$ \\
QGCA ($p{=}1$) & $5.79 \pm 0.55$ & $7.21 \pm 0.80$ \\
QGCA ($p{=}3$) & $6.07 \pm 0.78$ & $8.00 \pm 1.01$ \\
\bottomrule
\end{tabular}
\caption{Post-burn-in acceptance probabilities (\%) for different update methods of the MCMC algorithm analyzed in \Cref{sec:quantum_informed_markov_chain_monte_carlo} are shown. Single-spin flips (SSFs) and uniform-spin flips (USFs) are presented alongside QGCA flips, where QGCA denotes the CA guided by correlations derived from QAOA with respective depths $p$. The inverse temperatures are set to $\beta=3$ and $\beta=1$ for degrees 3 and 27, respectively.}
\label{tab:mcmc_acc_probs_post_percent}
\end{table}

\clearpage

\newgeometry{left=2cm,right=2cm,top=2cm,bottom=3cm}

\begin{figure*}[htbp]
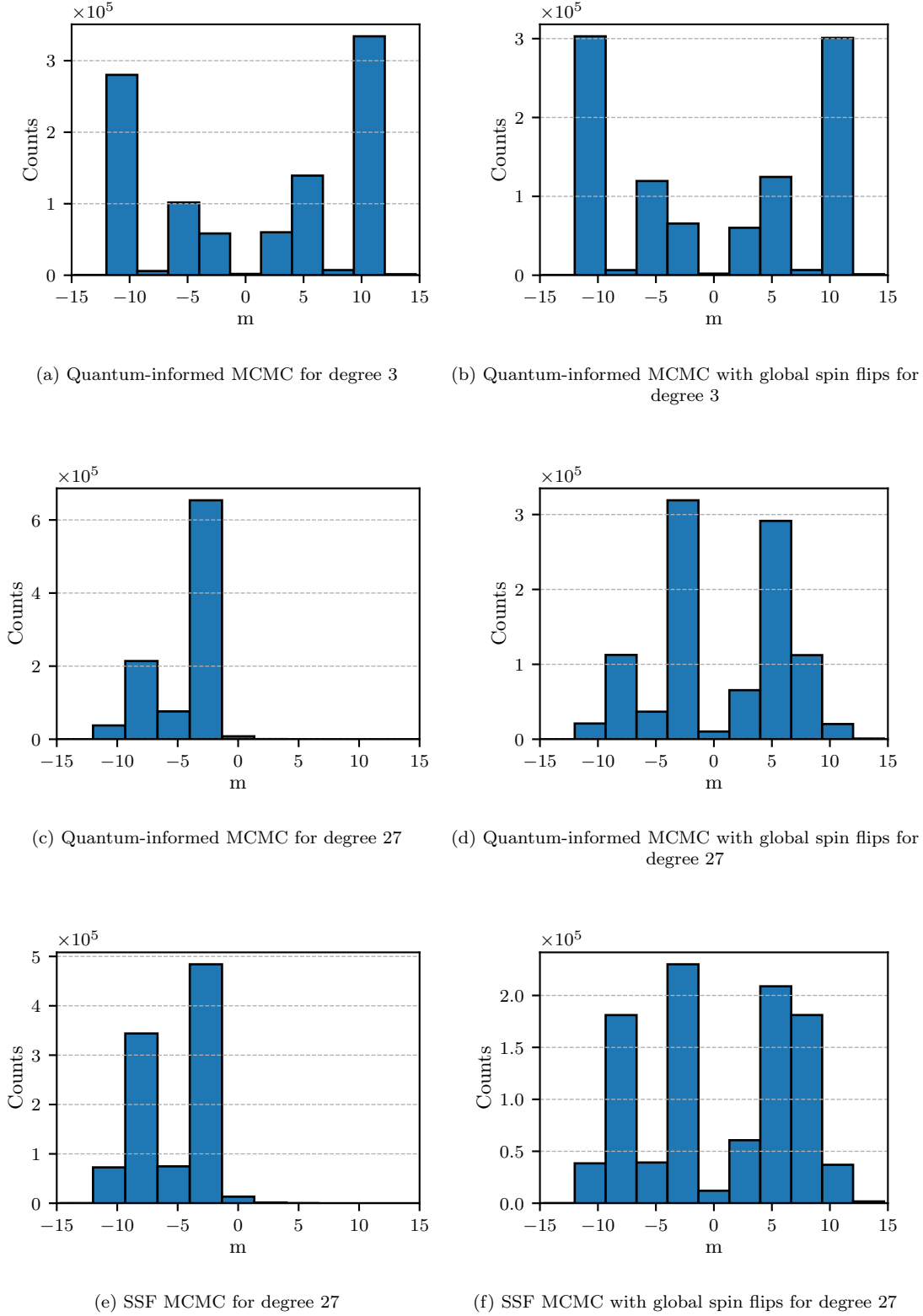

    \centering
    \subfloat[Quantum-informed MCMC for degree 3]{%
    \input{mcmc_magnetization_degree_3.pgf}%
        \label{fig:mcmc_magnetization_degree_3}%
    }
    \subfloat[Quantum-informed MCMC with global spin flips for degree 3]{%
    \input{mcmc_magnetization_degree_3_mix.pgf}%
        \label{fig:mcmc_magnetization_degree_3_mix}%
    }
    \par\medskip
    \subfloat[Quantum-informed MCMC for degree 27]{%
    \input{mcmc_magnetization_degree_27.pgf}%
        \label{fig:mcmc_magnetization_degree_27}%
    }
    \subfloat[Quantum-informed MCMC with global spin flips for degree 27]{%
    \input{mcmc_magnetization_degree_27_mix.pgf}%
        \label{fig:mcmc_magnetization_degree_27_mix}%
    }
    \par\medskip
    \subfloat[SSF MCMC for degree 27]{%
    \input{mcmc_magnetization_degree_27_SF.pgf}%
        \label{fig:mcmc_magnetization_degree_27_SF}%
    }
    \subfloat[SSF MCMC with global spin flips for degree 27]{%
    \input{mcmc_magnetization_degree_27_SF_mix.pgf}%
        \label{fig:mcmc_magnetization_degree_27_SF_mix}%
    }
    \caption{The post-burn-in magnetization histograms are shown for the quantum-informed MCMC algorithm in (a)–(d) and for the single-spin-flip (SSF) MCMC in (e) and (f). Due to the $\mathbb{Z}_2$ symmetry of the instances, the distributions are expected to be symmetric around zero magnetization. For degree-3 graphs, the inclusion of global spin-flip moves significantly improves this symmetry (cf.\ (a) and (b)). A similar effect is observed for degree-27 graphs; however, a slight asymmetry remains (cf.\ (c) and (d)), reflecting the increased sampling difficulty of these instances. A comparable behavior is observed for the SSF algorithm in (e) and (f), whose magnetization distributions closely resemble those obtained with the quantum-informed MCMC for degree-27 graphs.}
    \label{fig:mcmc_magnetization}
\end{figure*}

\clearpage

\bibliography{main}

\end{document}